%% file: main.tex
\documentclass[a4paper, 11pt]{article}
\usepackage{fullpage}
\usepackage{placeins}
\usepackage{amsthm,amsmath,amssymb}
\usepackage{fancyvrb}
\usepackage{graphicx}
\usepackage{booktabs}
\usepackage{multirow}
\usepackage{ragged2e}
\usepackage{amsmath}
\usepackage{array}
\usepackage{systeme}
\usepackage{listings}
\usepackage{amsfonts}
\usepackage{comment}
\usepackage{mathtools}
\usepackage{framed,enumitem} 
\usepackage{hyperref}
\hypersetup{colorlinks=true}
\hypersetup{linkcolor=black,anchorcolor=black,citecolor=black,urlcolor=black}
\usepackage{calrsfs}
\usepackage{bm}
\DeclareMathAlphabet{\pazocal}{OMS}{zplm}{m}{n}

\usepackage{amsthm}
\usepackage[ruled,vlined]{algorithm2e}
\usepackage{cite}
\usepackage{thmtools}

\newtheorem{definition}{Definition}
\newtheorem{observation}{Observation}
\newtheorem{theorem}{Theorem}
\newtheorem{corollary}{Corollary}
\newtheorem{lemma}{Lemma}

\title{Memoryless Algorithms for the Generalized $k$-server Problem on Uniform Metrics}
\date{}
\author{Dimitris Christou\thanks{National Technical University of Athens, Greece. \texttt{dimitrios.christou@hotmail.com, fotakis@cs.ntua.gr}.} 
\qquad
Dimitris Fotakis\footnotemark[1]
\qquad
Grigorios Koumoutsos\thanks{Universit\'{e} libre de Bruxelles, Belgium. \texttt{gregkoumoutsos@gmail.com}. Supported by Fonds de la Recherche Scientifique-FNRS Grant no MISU F 6001.}}

\newcommand{\MessageFromDimitris}[1]{\textcolor{red}{#1}}

\DeclareMathOperator{\ADV}{ADV}
\DeclareMathOperator{\ALG}{ALG}

\newcommand{\N}{\mathbb{N}}

\begin{document}
\maketitle
\justifying

\begin{abstract}
We consider the generalized $k$-server problem on uniform metrics. We study the power of memoryless algorithms and show tight bounds of $\Theta(k!)$ on their competitive ratio. In particular we show that the \textit{Harmonic Algorithm} achieves this competitive ratio and provide matching lower bounds. This improves the $\approx 2^{2^k}$ doubly-exponential bound of Chiplunkar and Vishwanathan for the more general setting of uniform metrics with different weights.   
\end{abstract}

\input{01-intro.tex}

\input{02-upper-bound.tex}
\input{03-lower-bound}
\input{04-extensions.tex}

\newpage
\bibliographystyle{plain}
\bibliography{references}

\newpage
\appendix
\input{A-markov-chains}

\clearpage
\input{B-Chiplunkar}
\clearpage
\input{C-n=2}

\end{document}

%% file: 01-intro.tex
\section{Introduction}

The $k$-server problem is one of the most fundamental and extensively studied problems in the theory of online algorithms. In this problem, we are given a metric space of $n$ points and $k$ mobile servers located at points of the metric space. At each step, a request arrives at a point of a metric space and must be served by moving a server there. The goal is to minimize the total distance travelled by the servers.

The $k$-server problem generalizes various online problems, most notably the paging (caching) problem, which corresponds to the $k$-server problem on uniform metric spaces. Paging, first studied in the seminal work of Sleator and Tarjan~\cite{ST85}, is well-understood: the competitive ratio is $k$ for deterministic algorithms and $H_k = \Theta(\log k)$ for randomized; those algorithms and matching lower bounds are folklore results for online algorithms~\cite{ST85,MS91,ACN00}.  

The $k$-server problem in general metric spaces is much deeper and intriguing. In a landmark result, Koutsoupias and Papadimitriou~\cite{KP95} showed that the \textit{Work Function Algorithm} (WFA)~\cite{KP95} is $(2k-1)$-competitive, which is almost optimal for deterministic algorithms since the competitive ratio is at least $k$~\cite{MMS90}. For randomized algorithms, it is believed that an $O(\log k)$-competitive algorithm is possible; despite several breakthrough results over the last decade~\cite{BBMN15,BCLLM18,BGMN19,Lee18}, this conjecture still remains open. 

\paragraph{Memoryless Algorithms:} One drawback of the online algorithms achieving the best-known competitive ratios for the $k$-server problem is that they are computationally inefficient. For example, the space used by the WFA is proportional to the number of different configurations of the servers, i.e., $\binom{n}{k}$, which makes the whole approach quite impractical.

This motivates the study of trade-offs between the competitive ratio and computational efficiency. A starting point in this line of research, is to determine the competitive ratio of \textit{memoryless algorithms:} a memoryless algorithm, decides the next move based solely on the current configuration of the servers and the given request. 

Memoryless algorithms for the $k$-server problem have been extensively studied (see e.g.,~\cite{BEY98,Kou09} for detailed surveys). The most natural memoryless algorithm is the \textit{Harmonic Algorithm}, which moves each server with probability inversely proportional to its distance from the requested point. It is known that its competitive ratio is $O(2^k \cdot \log k) $ and $\Omega(k^2)$~\cite{BG00}. It is conjectured that in fact the Harmonic Algorithm is $\frac{k (k+1)}{2} = O(k^2)$-competitive; this remains a long-standing open problem. For special cases such as uniform metrics and resistive metric spaces, an improved competitive ratio of $k$ can be achieved and this is the best possible for memoryless algorithms~\cite{CDRS93}. 

We note that the study of memoryless algorithms for the $k$-server problem is of interest only for randomized algorithms; it is easy to see that any deterministic memoryless algorithm is not competitive. Throughout this paper, we adopt the standard benchmark to evaluate randomized memoryless algorithms, which is comparing them against an \textit{adaptive online adversary}, unless stated otherwise. For a detailed discussion on the different adversary models and relations between them, see~\cite{BEY98,BBKTW94}.

\paragraph{The generalized $k$-server problem.} In this work, we focus on the generalized $k$-server problem, a far-reaching extension of the $k$-server problem, introduced by Koutsoupias and Taylor~\cite{KT04}. Here, each server $s_i$ lies in a different metric space $M_i$ and a request is a tuple $(r_1,\dotsc,r_k)$, where  $r_i \in M_i$; to serve it, some server $s_i$ should move to point $r_i$. The standard $k$-server problem is the very special case where all metric spaces are identical, i.e., $M_i = M$ and all requests are of the form $(r,r,\dotsc,r)$. Other well-studied special cases of the generalized $k$-server problem are (i) the weighted $k$-server problem~\cite{FR94,BEK17}, where all metrics are scaled copies of a fixed metric $M$, i.e., $M_i = w_i M$ and all requests of the form $r = (r,r,\dotsc,r)$ and (ii) the CNN problem~\cite{KT04,iw01}, where all metrics $M_i$ are real lines.

\paragraph{Previous Work.} The generalized $k$-server problem has a much richer structure than the classic $k$-server problem and is much less understood. For general metric spaces, no $f(k)$- competitive algorithms are known, except from the special case $k=2$~\cite{SSP03,SS06,Sit14}. For $k \geq 3$, competitive algorithms are known only for the following special cases:

\begin{enumerate}
    \item Uniform Metrics: All metric spaces $M_1,\dotsc, M_k$ are uniform (possibly with different number of points), with the same pairwise distance, say 1.
    \item Weighted Uniform Metrics: All metrics are uniform, but they have different weights; the cost of moving in metric $M_i$ is $w_i$. 
\end{enumerate}

Perhaps surprisingly, those two cases are qualitatively very different. For deterministic algorithms Bansal et. al.~\cite{BEKN18} obtained algorithms with (almost) optimal competitive ratio. For uniform metrics their algorithm is $(k \cdot 2^k)$-competitive, while the best possible ratio is at least $2^{k} -1 $~\cite{KT04}. For weighted uniform metrics, they obtained a $2^{2^{k+3}}$-competitive algorithm (by extending an algorithm of Fiat and Ricklin~\cite{FR94} for weighted $k$-server on uniform metrics), while the lower bound for the problem is $2^{2^{k-4}}$~\cite{BEK17}. 

\medskip

\noindent \textit{Memoryless Algorithms for Generalized $k$-server:} Recently Chiplunkar and Vishnawathan~\cite{ChV20} studied randomized memoryless algorithms in weighted uniform metrics. They showed tight doubly exponential ($\approx 1.6^{2^{k}}$) bounds on the competitive ratio. Interestingly, the memoryless algorithm achieving the optimal bound in this case is different from the Harmonic Algorithm.

Since the weighted uniform case seems to be much harder than the uniform case, it is natural to expect that a better bound can be achieved by memoryless algorithms in uniform metrics. Moreover, in weighted uniform metrics the competitive ratios of  deterministic algorithms (with memory) and randomized memoryless algorithms are essentially the same. Recall that a similar phenomenon occurs for the paging problem (standard $k$-server on uniform metrics) where both deterministic and randomized memoryless algorithms have a competitive ratio of $k$. Thus, it is natural to guess that for uniform metrics, a competitive ratio of order $2^k$ (i.e., same as the deterministic competitive ratio) can be achieved by memoryless algorithms.

\subsection{Our Results}

In this work we study the power of memoryless algorithms for the generalized $k$-server problem in uniform metrics and we determine the exact competitive ratio by obtaining tight bounds.

First, we determine the competitive ratio of the Harmonic Algorithm on uniform metrics.

\begin{theorem}
\label{thm:upper_bound}
The Harmonic Algorithm for the generalized $k$-server problem on uniform metrics is $(k \cdot \alpha_k)$-competitive, where $\alpha_k$ is the solution of the recursion $\alpha_k = 1 + (k-1) \alpha_{k-1}$, with $\alpha_1 = 1$. 
\end{theorem}

It is not hard to see that $\alpha_k = \Theta((k-1)!)$, therefore the competitive ratio of the Harmonic Algorithm is $O(k!)$. This shows that indeed, uniform metrics allow for substantial improvement on the performance compared to weighted uniform metrics where there is a doubly-exponential lower bound.

To obtain this result, we analyse the Harmonic Algorithm using Markov Chains and random walks, based on the \textit{Hamming distance} between the configuration of the algorithm and the adversary, i.e., the number of metric spaces where they have their servers in different points. Based on this, we then provide a proof using a potential function, which essentially captures the expected cost of the algorithm until it reaches the same configuration as the adversary. The proof is in Section~\ref{section:harmonic_upper_bound}.

Next we show that the upper bound of Theorem~\ref{thm:upper_bound} is tight, by providing a matching lower bound.

\begin{theorem}
\label{thm:gen_lower_bound}
The competitive ratio of any randomized memoryless algorithm for the generalized $k$-server problem on uniform metrics is at least $k \cdot \alpha_k$. 
\end{theorem}

Here the analysis differs, since the Hamming distance is not the right metric to capture the ``distance'' between the algorithm and the adversary: assume that all their servers are at the same points, except one, say server $s_i$. Then, in the next request, the algorithm will reach the configuration of the adversary with probability $p_i$; clearly, if $p_i$ is large, the algorithm is in a favourable position, compared to the case where $p_i$ is small.

This suggests that the structure of the algorithm is not solely characterized by the number of different servers (i.e., Hamming distance) between the algorithm and the adversary, but also the labels of the servers matter. For that reason, we need to focus on the subset of different servers, which gives a Markov Chain on $2^k$ states. Unfortunately, analyzing such chains in a direct way can be done only for easy cases like $k=2$ or $k=3$. For general values of $k$, we find an indirect way to characterize the solution of this Markov Chain. A similar approach was taken by Chiplunkar and Vishwanathan~\cite{ChV20} for weighted uniform metrics; we use some of the properties they showed, but our analysis differs since we need to make use of the special structure of our problem to obtain our bounds. 

In fact, we are able to show that any memoryless algorithm other than the Harmonic has competitive ratio strictly larger than $k \cdot \alpha_k$. We describe the details in Section~\ref{sec:gen_lb}.

On the positive side, our results show that improved guarantees can be achieved compared to the weighted uniform case. On the other hand, the competitive ratio of memoryless algorithms ($\Theta(k!)$) is asymptotically worse than the deterministic competitive ratio of $2^{O(k)}$. This is somewhat surprising, since (as discussed above) in most uniform metric settings of $k$-server and generalizations, the competitive ratio of deterministic algorithms (with memory) and randomized memoryless is (almost) the same.

\subsection{Notation and Preliminaries}

\paragraph{Memoryless Algorithms.} A memoryless algorithm for the generalized $k$-server problem receives a request $r = (r_1,\dotsc,r_k)$ and decides which server to move based only on its current configuration $q = (q_1, \dotsc, q_k)$ and $r$. For the case of uniform metrics, a memoryless algorithm is fully characterized by a probability distribution $p = (p_1,\dotsc,p_k)$; whenever it needs to move a server, it uses server $s_i$ of metric $M_i$ with probability $p_i$. Throughout the paper we assume for convenience (possibly by relabeling the metrics) that given a memoryless algorithm we have that $p_1 \geq p_2 \geq \dotsc \geq p_k$. We also assume that $p_i>0$ for all $i$; otherwise it is trivial to show that the algorithm is not competitive.

\paragraph{The Harmonic Algorithm.} In the context of generalized $k$-server on uniform metrics, the Harmonic Algorithm is a memoryless algorithm which moves at all metric spaces with equal probability, i.e., $p_i = 1/k$, for all $i \in [k]$.

\paragraph{The harmonic recursion.} We now define the recursion that will be used to get the competitive ratio of the harmonic algorithm, which we call the \textit{harmonic recursion} and do some basic observations that will be useful throughout the paper. 


\begin{definition}[Harmonic recursion]
The harmonic recursion $\alpha_\ell\in\mathbb{N}$ satisfies the recurrence $\alpha_\ell = 1 + (\ell-1)\alpha_{\ell-1}$ for $\ell>1$, and $\alpha_1=1$. 
\end{definition}


Based on the definition, we make the following observation:

\begin{observation}\label{ob:harmonic_recursion_increasing}
The harmonic recursion is \textit{strictly} increasing, i.e., $\alpha_{\ell+1}>\alpha_\ell$ for any $\ell\in\N$.
\end{observation}

Also it is easy to show that $\alpha_\ell$ has a closed form, given by 

\begin{equation}\label{eq:harmonic_recursion_closed_form}
\alpha_\ell= (\ell-1)!\sum_{i=0}^{\ell-1}\frac{1}{i!}.
\end{equation}
\noindent Based on this closed form, we get the following:

\begin{observation}\label{ob:harmonic_recursion_magnitude}
For any $\ell\in\N$, it holds that $(\ell-1)! \leq \alpha_\ell \leq e(\ell-1)! $.
\end{observation}
\noindent This observation also shows that for any $\ell\in\N$, we have $\alpha_\ell = \Theta((\ell-1)!)$.

%% file: 02-upper-bound.tex
\section{Upper Bound}\label{section:harmonic_upper_bound}

In this section we prove Theorem~\ref{thm:upper_bound}. More precisely, we use a potential function argument to show that for any request sequence, the expected cost of the Harmonic Algorithm is at most $k\cdot\alpha_k$ times the cost of the adversary.

\paragraph{Organization.} In Section~\ref{subsection:selecting_potential}, we define a potential between the Harmonic Algorithm's and the adversary's configurations that is inspired by random walks on a special type of Markov Chains~\cite{prob:2003} we refer to as the ``\textit{Harmonic Chain}''. The required background of Markov Chains is presented in Appendix~\ref{app:markov_chains}. Then, in Section~\ref{subsection:potential_analysis} we will use this potential to prove the upper bound of Theorem~\ref{thm:upper_bound} with a standard potential-based analysis.



\subsection{Definition of the Potential Function}\label{subsection:selecting_potential}

We begin by presenting the intuition behind the definition of our potential function. Our first observation is that since (i) the metrics are uniform with equal weights and (ii) the Harmonic Algorithm does not distinguish between metrics since it has equal probabilities $\frac{1}{k}$, it makes sense for the potential between two configurations $p,q\in [n]^k$ to depend only on their Hamming distance and not on the labels of their points. In order to come up with an appropriate potential, we need to understand how the Hamming distance between the Harmonic Algorithm's and the adversary's configurations evolves over time. 

Imagine that the adversary moves to an ``optimal'' configuration of his choice and then it serves requests until the Harmonic Algorithm reaches this configuration as well. Since the adversary must serve all the requests using a server from its configuration, we know that for each request $r = (r_1,\dotsc,r_k)$, at least one of the requested points $r_i$ should coincide with the $i$-th server of the adversary.  In that case, with probability $\frac{1}{k}$ the Harmonic Algorithm moves in metric $M_i$, thus it decreases his Hamming distance from the adversary by 1. On the other hand, assume that $\ell$ servers of the algorithm coincide with the ones of the adversary. Then, with probability $\frac{\ell}{k}$ it would increase its Hamming distance from the optimal configuration by 1. This shows that the evolution of the Hamming distance between the Harmonic Algorithm's and the adversary's configurations is captured by a random walk on the following Markov Chain that we refer to as the \textit{Harmonic Chain}.
\begin{figure}[ht]
    \centering
    \includegraphics[width=\textwidth]{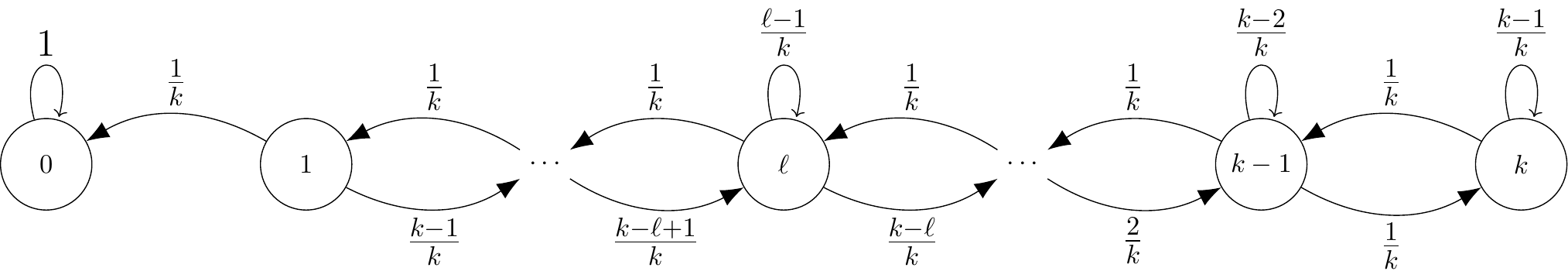} 
    \caption{The Harmonic Chain - Here, the states of the chain denote the Hamming distance between the configurations of the Harmonic Algorithm and the adversary.}
\end{figure}
\FloatBarrier

While not necessary for the definition of the potential, a formal definition of the Harmonic Chain is included in Appendix~\ref{app:markov_chains}. In the scenario we described above, the expected movement cost of the Harmonic Algorithm until it reaches the adversary's configuration with an initial Hamming distance of $\ell$ would be $\mathbb{E}[N|X_0=\ell]$ where $N$ denotes a random variable defined as $N=min_{\tau \geq 0}\{X_\tau = 0\}$ and $X_t$ denotes the state of the Harmonic Chain at time $t$. In the literature, this quantity is known as the \textit{Expected Extinction Time} (ETT) ~\cite{prob:2003} of a Markov Chain and we use $h(\ell)$ to denote it. Intuitively, $h(k)$ should immediately give an upper bound on the competitive ratio of the Harmonic Algorithm.

We study the Harmonic Chain and prove the following Theorem:

\begin{theorem}\label{thm:harmonic_EET}
For any initial state $\ell\in\{0,1,\dotsc,k\}$, the EET of the Harmonic Chain is given by 
$$h(\ell) = k\sum_{i=k-\ell+1}^{k}\alpha_{i}$$
\end{theorem}
\begin{proof}
By using conditional probabilities on the EET of the Harmonic Chain, we get that for any $\ell\in\{1,2,\dotsc ,k\}$, $$h(\ell)=1+\frac{1}{k}h(\ell-1) + \frac{k-\ell}{k}h(\ell+1) + (1-\frac{1}{k}-\frac{k-\ell}{k})h(\ell).$$ 
This yields a second-order recurrence relation we need to solve for $h(\ell)$. A formal proof is given in Appendix~\ref{app:markov_chains}, where we derive the Theorem from the EET of the more general class of Markov Chains called Birth-Death Chains.
\end{proof}

From Theorem~\ref{thm:harmonic_EET} and Observations~\ref{ob:harmonic_recursion_increasing},~\ref{ob:harmonic_recursion_magnitude} we immediately get that $h(\ell)$ is (strictly) increasing and $h(\ell)=\Theta(\ell!) \;\forall \ell\in\{1,\dotsc ,k\}$. Furthermore, we make the following observation: 

\begin{observation}\label{cor:h(n)/n}
For any $\ell\in \{1,2,\dotsc,k\}$, $\frac{h(\ell)}{\ell} \leq h(1)$ in the Harmonic Chain, with the equality holding only for $\ell=1$.
\end{observation}
\begin{proof}
Fix any $\ell\in \{2,3,\dotsc,k\}$. Then: 
$$\frac{h(\ell)}{\ell} = \frac{k\alpha_k + k\sum_{i=k-\ell}^{k-2}\alpha_{i+1}}{\ell} 
    <  \frac{k\alpha_k + k(\ell-1)\alpha_{k-1}}{\ell} 
    < \frac{k\alpha_k + (\ell-1)k\alpha_k}{\ell} = k\alpha_k = h(1).$$
where both inequalities hold from Observation~\ref{ob:harmonic_recursion_increasing} which states that $\alpha_i$ is \textit{strictly} increasing.
\end{proof}

Suppose that the adversary moves $\ell$ servers whenever the algorithm reaches its configuration and then it doesn't move until the algorithm reaches its new configuration. Intuitively, the competitive ratio would be $\frac{h(\ell)}{\ell}$ which is maximized for $\ell=1$ by Observation~\ref{cor:h(n)/n}. This means that $h(1)=k\cdot\alpha_k$ is an upper bound for the competitive ratio of the Harmonic Algorithm. While this intuition is very important, it is not enough to formally prove Theorem~\ref{thm:upper_bound}. However, motivated by it, we will define the potential between two configurations of Hamming distance $\ell$ as $h(\ell)$. Formally,

\begin{definition}[Potential Function]
The potential between two configurations $p,q\in [n]^k$ is defined as
$$\phi(p,q) = h(d_H(p,q)).$$
\end{definition}

\subsection{Bounding the Competitive Ratio}\label{subsection:potential_analysis}

In this section, we will prove the upper bound of Theorem \ref{thm:upper_bound} by using the potential we defined in Section \ref{subsection:selecting_potential}. Fix any request sequence $\bar{r}=[r^1,\dotsc,r^T]$ for any $T\in\N$ such that $r^t\in [n]^k\;\forall t\in [T]$. Let $q^t\in [n]^k$ be the configuration of the Harmonic Algorithm and $\pazocal{A}^t \in [n]^k$ the configuration of the adversary after serving request $r^t$. Also, let $q^0=\pazocal{A}^0$ be the initial configuration of the instance. We will prove that when the adversary moves $x$ servers the increase in potential is at most $k \cdot \alpha_k \cdot x$ and when the Harmonic Algorithm moves one server, the expected decrease in potential is at least $1$. Then, using these properties, we will prove Theorem \ref{thm:upper_bound}.

To simplify the analysis, we make the following observation for the potential function.
\begin{observation}\label{obs:Dh}
For any $\ell,\ell' \in \{0,1,\dotsc, k\}$ such that $\ell < \ell'$ it holds that 
$$h(\ell')-h(\ell) = k\sum_{i=\ell}^{\ell'-1}\alpha_{k-i}$$
\end{observation}
\begin{proof} By telescoping we have
$$h(\ell')-h(\ell) = \sum_{i=\ell}^{\ell'-1}(h(i+1)-h(i)) =\sum_{i=\ell}^{\ell'-1}(k\sum_{j=k-i}^k\alpha_j - k\sum_{j=k-i+1}^k\alpha_j) = k\sum_{i=\ell}^{\ell'-1}\alpha_{k-i}$$
where the second equality holds by the definition of the potential. \end{proof}

Using this observation, we are now ready to prove the following lemmata:
\begin{lemma}[Adversary Moves]\label{lemma:opt_moves}
For any $t\in\{1,\dotsc,T\}$ it holds that
$$\phi(q^{t-1},\pazocal{A}^t)-\phi(q^{t-1},\pazocal{A}^{t-1}) \leq k \cdot \alpha_k \cdot  d_H(\pazocal{A}^t,\pazocal{A}^{t-1}).$$
\end{lemma}
\begin{proof}
Let $\ell^{t-1}=d_H(q^{t-1},\pazocal{A}^{t-1})$ and $\ell^t=d_H(q^{t-1},\pazocal{A}^{t})$. Clearly, $\ell^{t-1},\ell^{t}\in \{0,1,\dotsc, k\}$. Since the potential $h(\ell)$ is strictly increasing on $\ell$, if $\ell^{t}\leq \ell^{t-1}$ then this means that the adversary's move didn't increase the potential and then the Lemma follows trivially. Thus, we only need to prove the Lemma for $0 \leq \ell^{t-1} < \ell^{t} \leq k$. We have:

\begin{equation}\label{eq:adv_moves}
h(\ell^t)-h(\ell^{t-1}) = k\sum_{i=\ell^{t-1}}^{\ell^t -1}\alpha_i \leq (\ell^{t}-\ell^{t-1})k\alpha_k    
\end{equation}
where the equality is given from Observation~\ref{obs:Dh} and the inequality from the fact that the recursion $\alpha_\ell$ is increasing. Thus, we have proven that $\phi(q^{t-1},\pazocal{A}^t)-\phi(q^{t-1},\pazocal{A}^{t-1}) \leq (\ell^{t}-\ell^{t-1})\cdot k\cdot\alpha_k$. To conclude the proof of the Lemma, by the triangle inequality of the Hamming distance we have
$$d_H(q^{t-1},\pazocal{A}^{t-1}) + d_H(\pazocal{A}^{t-1},\pazocal{A}^{t}) \geq  d_H(q^{t-1},\pazocal{A}^{t})$$
which gives $\ell^{t}-\ell^{t-1} \leq d_H(\pazocal{A}^{t-1},\pazocal{A}^t)$. Combined with~\eqref{eq:adv_moves}, we get the Lemma.
\end{proof}

\begin{lemma}[Harmonic Moves]\label{lemma:alg_moves}
For any $t\in\{1,\dots,T\}$ it holds that
$$\mathbb{E}[\phi(q^{t-1},\pazocal{A}^t)-\phi(q^{t},\pazocal{A}^{t})] \geq d_H(q^{t-1},q^t).$$
\end{lemma}
\begin{proof}
If the Harmonic Algorithm serves the request, then $q^t=q^{t-1}$ and the Lemma follows trivially. Otherwise, by definition, it moves to a configuration $q^t$ such that $d_H(q^{t-1},q^t)=1$. Let $\ell^{t-1}=d_H(q^{t-1},\pazocal{A}^{t})$ and $\ell^t=d_H(q^{t},\pazocal{A}^{t})$. Also, let $C = |\{i: \pazocal{A}^t_i = r^t_i\}|$, i.e., the number of the adversary's servers that could serve the current request. By definition, $\pazocal{A}^t$ must serve $r^t$ which gives $C\geq 1$. Furthermore, $q^{t-1}$ doesn't serve the request but $\pazocal{A}^t$ does, and thus $\ell^{t-1}\geq 1$.

Recall that the Harmonic Algorithm randomly moves at a metric with equal probabilities in order to serve a request. If it moves in any of the $C$ metrics where the adversary serves the request, we get $\ell^{t}=\ell^{t-1}-1$ and the potential decreases with probability $\frac{C}{k}$. If it moves on any of the $k-\ell^{t-1}$ metrics where $a^t_i=q^{t-1}_i$, we get $\ell^{t}=\ell^{t-1}+1$ and the potential increases with probability $\frac{k-\ell^{t-1}}{k}$. In any other case, we have $\ell^{t}=\ell^{t-1}$ and the potential doesn't change. To simplify the notation, we define $j\in\{1,\dotsc,k\}$ as $j=k - \ell^{t-1} +1$. We have:

\begin{align*}
\mathbb{E}[\phi(q^{t-1},\pazocal{A}^t)-\phi(q^{t},\pazocal{A}^{t})] 
&= \mathbb{E}[h(\ell^{t-1})-h(\ell^{t})] \\
&= \frac{C}{k}(h(\ell^{t-1})-h(\ell^{t-1}-1)) + \frac{k-\ell^{t-1}}{k}(h(\ell^{t-1})-h(\ell^{t-1}+1)) \\
&= C\alpha_{j} - (j-1)\alpha_{j-1}\\
&= C\alpha_{j} - (\alpha_{j} -1) = (C-1)\alpha_j + 1 \\
&\geq 1 = d_h(q^{t-1},q^t)
\end{align*}
where the first equality follows from the definition of the potential, the second equality from the possible changes in the Hamming distance between the algorithm and the adversary, the third equality follows from Observation~\ref{obs:Dh} and the definition of $j$, the fourth equality follows from the definition of the recursion $\alpha_\ell$ and the inequality follows from $C\geq 1$.
\end{proof}

\paragraph{Proof of Theorem~\ref{thm:upper_bound}} We are now ready to prove Theorem \ref{thm:upper_bound}. By combining lemmata \ref{lemma:opt_moves} and \ref{lemma:alg_moves}, we get that for any $t\in\{1,\dotsc, T\}$, the expected difference in potential is $$\mathbb{E}[\Delta\phi^t]= \mathbb{E}[\phi(q^t,\pazocal{A}^t)-\phi(q^{t-1},\pazocal{A}^{t-1})] \leq k\alpha_kd_H(\pazocal{A}^t,\pazocal{A}^{t-1}) - d_H(q^{t-1},q^t)$$

Now, let $\ADV = \sum_{t=1}^Td_H(\pazocal{A}^t,\pazocal{A}^{t-1})$ be used to denote the total cost of the adversary and $\ALG = \sum_{t=1}^Td_H(q^t,q^{t-1})$ be used to denote the expected cost of the Harmonic Algorithm. Summing over all $t\in\{1,2,\dotsc, T\}$ we finally get
$$\sum_{t=1}^T\Delta\phi^t = \phi(q^T,\pazocal{A}^T)-\phi(q^0,\pazocal{A}^0) \leq k\alpha_k \cdot \ADV - \ALG$$
and since $\pazocal{A}^0=q^0$ (i.e., $\phi(q^0,\pazocal{A}^0)=0$) and $\phi(q^T,\pazocal{A}^T)\geq 0$, we get that $\ALG \leq k \cdot \alpha_k \cdot \ADV$, which concludes the proof of Theorem \ref{thm:upper_bound}.

%% file: 03-lower-bound.tex
\section{Lower Bound}\label{sec:gen_lb}

In this section we prove Theorem~\ref{thm:gen_lower_bound}. More precisely, we construct an adversarial request sequence against any memoryless algorithm and prove that its competitive ratio is lower bounded by the solution of a linear system of $2^k$ equations. Since solving this system directly is possible only for easy cases like $k=2$ or $k=3$, we show how to get a lower bound for the solution (similarly to the approach taken by Chiplunkar and Vishwanathan~\cite{ChV20} for weighted uniform metric spaces) and thus the competitive ratio of any memoryless algorithm.

\paragraph{Organization.} In Section~\ref{subsection:general_lb_instance} we formally define the adversarial request sequence and the intuition behind it. In Section~\ref{subsection:general_lb_equations} we state the linear system of equations that our request sequence results to and prove a lower bound on its solution. This leads to the proof of Theorem~\ref{thm:gen_lower_bound}.

\subsection{Constructing the adversarial instance}\label{subsection:general_lb_instance}

Before we state the adversarial instance, it is useful to give the intuition behind it. It is natural to construct an adversary that moves only when it has the same configuration with the algorithm. 

In fact, we construct an adversary that moves in only one metric space: the one that the algorithm uses with the smallest probability (ties are broken arbitrarily). Recall that in the analysis of the harmonic algorithm from Section~\ref{section:harmonic_upper_bound}, the competitive ratio is also maximized when in each ``phase'' the adversary starts with only one different server than the algorithm and does not move until the configurations (of algorithm and adversary) match (Observation~\ref{cor:h(n)/n}).

Let $\ALG$ be any online algorithm and $\ADV$ be the adversary.
Consider a ``phase'' to be a part of the request sequence where in the beginning the configurations of $\ALG$ and $\ADV$ coincide and it ends when $\ALG$ matches the configuration of $\ADV$. Since $\ADV$ must serve all requests, in each request $r$ one point $r_i$ is such that $a_i = r_i$; we say that the $i$th position of $\ADV$ is revealed in such a request. Thus every request will reveal to the algorithm exactly one of the positions of the adversary's servers in some metric space $M_i$. The main idea behind our lower bound instance is that, in each request, out of the metric spaces that servers of $\ALG$ and $\ADV$ differ, we reveal to the algorithm the position of $\ADV$ in the metric that $\ALG$ serves with the highest probability; this implies that whenever $\ALG$ and $\ADV$ differ by only one server, this will be in metric $M_k$. Intuitively, this way we exploit best the ``assymetries'' in the distribution of $\ALG$ (this is formalized in Lemma~\ref{lemma:monotonicity}).

\paragraph{The instance.} Recall that any memoryless algorithm for the generalized $k$-server problem on uniform metric spaces is fully characterized by a probability distribution $p=[p_1,p_2,\dotsc,p_k]$ over the $k$-metric  spaces $M_1,M_2,\dots , M_k$. W.l.o.g., we can assume that $p_1\geq p_2 \geq \dots \geq p_k$. Let $q^t,\pazocal{A}^t$ be used to denote the configurations of the algorithm and the adversary after serving request $r^t$ respectively. Also, let $q^0=\pazocal{A}^0$ be used to denote the initial configuration of both the algorithm and the adversary. We will now construct the request sequence. For $t=1,2,\dots,T$:

\begin{enumerate}
    \item Observe $q^{t-1}$, i.e., the algorithm's current configuration.
    \item If $q^{t-1}=\pazocal{A}^{t-1}$, then:
        \subitem $\pazocal{A}^t = [q^0_1 ,q^0_2, \dots , q^0_{k-1}, Z]$ for any $Z\in [n]$ such that $Z\neq \pazocal{A}^{t-1}_k$ and $Z\neq q^{t-1}_k$.

    otherwise:
        \subitem $\pazocal{A}^t = \pazocal{A}^{t-1}$.

    \item Determine $m = \min(\{j: q^{t-1}_j \neq \pazocal{A}^t_j\})$.
    \item Pick any $r^t\in [n]^k$ such that $r^t_m = \pazocal{A}^t_m$ and $r^t_j\neq \pazocal{A}^t_j$, $r^t_j \neq q^{t-1}_j \;\forall j\in [k]\setminus \{m\}$.
\end{enumerate}

Note that for steps 2 and 4, we need to have at least $n \geq 3$ points in order to pick a point that isn't occupied by neither the algorithm's nor the adversary's servers. As we explain in Section~\ref{sec:ext}, this is a necessary requirement; if all metrics have $n=2$ points, then the competitive ratio of the Harmonic Algorithm is $O(2^k)$ and therefore a lower bound of order $k!$ is not possible.

As an example of our instance, for $k=4$, let $\pazocal{A}^{t-1}=[0,0,0,0]$ and $q^{t-1}=[1,0,0,1]$ for some $t$. Clearly, the algorithm and the adversary have different servers in metric $M_1$ and $M_4$. From step 3, $m=\min(1,4)=1$, i.e., $M_1$ is the metric space that the algorithm serves with highest probability out of the metric spaces that it and the adversary have their servers in different points. Then, from step 4, $r^t=[0,2,2,2]$ (actually, the selection of the last three coordinates is arbitrary as long as neither the algorithm nor the adversary have their server on this point).

Notice that $\ADV$ moves one server in metric space $M_k$ whenever it has the same configuration with $\ALG$. On the other hand, $\ALG$ never serves request $r^t$ with configuration $q^{t-1}$ and thus moves at every time step. This means that the competitive ratio of $\ALG$ is lower bounded by the (expected) number of requests it takes for it to reach configuration of $\ADV$.



\subsection{Proving the Lower Bound}\label{subsection:general_lb_equations}

\paragraph{Our approach.} We define the \textit{state} of the algorithm at time $t$ as $S^t = \{i:q^t_i \neq \pazocal{A}^t_i\}$, i.e., the subset of metric spaces with different servers between the algorithm and the adversary. In this context, $h(S)$ is used to denote the expected number of requests it takes for the algorithm to reach the adversary's configuration, i.e. state $\emptyset$, starting from some state $S\subseteq [k]$. From the request sequence we defined, $h(\{k\})$ is a lower bound for the competitive ratio of any memoryless algorithm.

By observing how the state $S$ of the algorithm (and by extension $h(S)$) evolves under the request sequence, we can write down a linear system of $2^k$ equations on the $2^k$ variables $h(S)\;\forall S\subseteq [k]$. In fact, these equations give the EET of a random walk in a Markov Chain of $2^k$ states. We then prove a lower bound on $h(\{k\})$ and thus the competitive ratio of any memoryless algorithm. Notice that for the given instance, if we were analyzing the Harmonic Algorithm, then the Hamming distance between it and the adversary would be captured by the Harmonic Chain and we would immediately get that $h(\{k\})=k\cdot\alpha_k$.

\paragraph{Analysis.} Fix any two different configurations $q,\pazocal{A}$ for the algorithm and the adversary that are represented by state $S=\{i:q_i\neq\pazocal{A}_i\}\neq \emptyset$ with $\min(S)=m$. Then, we know that for the next request $r$ we have constructed it holds that $r_m=\pazocal{A}_m\neq q_m$ and $r_j\neq\pazocal{A}_j\neq q_j \neq r_j$ for any $j\in [k]\setminus\{m\}$. Recall that the memoryless algorithm will randomly move to some state $M_j$ and move to a different configuration $q'=[q_1,\dots, q_{j-1},r_j,q_{j+1},\dots,q_k]$ that is captured by state $S'$. We distinguish between the following three cases:
\begin{enumerate}
    \item If $j\notin S$, then this means that $q_j=\pazocal{A}_j$ and $q'_j = r_j \neq \pazocal{A}_j$ and thus $S'=S\cup\{j\}$.
    \item If $j=m$, then $q_j\neq \pazocal{A}_j$ and $q'_j = \pazocal{A}_j = r_m$ and thus $S'=S\setminus \{m\}$.
    \item If $j\in S\setminus\{m\}$ then $q_j\neq \pazocal{A}_j$ and $q'_j\neq \pazocal{A}_j$ and thus $S'=S$.
\end{enumerate}

Since $h(S)$ denotes the expected number of steps until the state of the algorithm becomes $\emptyset$ starting from $S$, from the above cases we have that for any state $S\neq \emptyset$:
\begin{equation*}
    h(S) = 1 + p_m \cdot h(S\setminus\{m\}) + \sum_{j\notin S}p_j \cdot h(S\cup \{j\}) + \sum_{j\in S\setminus\{m\}}p_j \cdot h(S), \;\; m = \min (S).
\end{equation*}

Combined we the fact that obviously $h(\emptyset)=0$ and $\sum_{j=1}^kp_j =1$, we get the following set of $2^k$ linear equations with $2^k$ variables:

\begin{equation}\label{system:lb}
\left\{
\begin{array}{c}
     h(\emptyset) = 0\\
     p_m(h(S)- h(S\setminus \{m\})) = 1 +\sum_{j\notin S}p_j(h(S\cup \{j\})-h(S)), \quad \forall S\neq\emptyset , m=\min (S)\\
\end{array} 
\right\} 
\end{equation}

Normally, we would like to solve this linear system to compute $h(\{k\})$ and this would be the proven lower bound for the memoryless algorithm. However, even for $k=4$ it is hopeless to find closed form expressions for the solutions of this system. Interestingly, similar equations were studied by Chiplunkar and Vishnawathan~\cite{ChV20} for the weighted uniform metric case. In their study, they showed a monotonicity property on the solutions of their linear system that directly transfers to our setting and is stated in Lemma~\ref{lemma:monotonicity} below. Using this, combined with the special structure of our problem, we show how to derive a lower bound of $k\cdot\alpha_k$ for $h(\{k\})$ instead of solving~\eqref{system:lb} to directly compute it.

\begin{lemma}\label{lemma:monotonicity}
For any $S\subseteq [k]$ with $i,j\in S$ such that $i<j$ (and thus $p_i\geq p_j$), the solutions of linear system \eqref{system:lb} satisfy
$$h(S)-h(S\setminus\{j\}) \geq \frac{p_i}{p_j}(h(S)-h(S\setminus\{i\}))$$
\end{lemma}

The proof is deferred to Appendix~\ref{app:Chiplunkar}.
Let us first see the intuition behind the inequality of Lemma~\ref{lemma:monotonicity}. Let $S$ be the subset of metric spaces where the servers of $\ALG$ and $\ADV$ occupy different points: then, in the next move, the expected time to match $\ADV$ decreases the most, if $\ALG$ matches first the $j$th server of the adversary (i.e., the ``state'' changes from $S$ to $S \setminus \lbrace j \rbrace$) where $j$ is the metric with the smallest the probability $p_j$. This explains why in our adversarial instance we choose to reveal to $\ALG$ the location of $\ADV$ in the metric it serves with the highest probability: this makes sure that the decrease in the expected time to reach $\ADV$ is minimized.

Using Lemma~\ref{lemma:monotonicity}, we can now prove the following:
\begin{lemma}\label{lemma:induction_property}
For any $S\subseteq [k]$ with $S\neq\emptyset$ and $i\in S$, the solutions of linear system \eqref{system:lb} satisfy
$$p_i(h(S)-h(S\setminus \{i\})) \geq 1 + \sum_{j\notin S}p_j(h(S\cup\{j\})-h(S))$$
\end{lemma}
\begin{proof}
Fix any non-empty set $S\subseteq [k]$ and any $i\in S$. Let $m=\min(S)\leq i$. Then, by Lemma \ref{lemma:monotonicity} we have
$$p_i(h(S)-h(S\setminus \{i\})) \geq p_m (h(S)-h(S\setminus \{m\}))$$
Since $m=\min(S)$, and we study the solution of linear system \eqref{system:lb}, we have
$$p_m (h(S)-h(S\setminus \{m\})) = 1 + \sum_{j\notin S}p_j(h(S\cup\{j\})-h(S))$$
and the lemma follows.
\end{proof}

We are now ready to prove the main theorem of this section.
\begin{theorem}\label{th:h(k)_lb}
The solution of linear system \eqref{system:lb} satisfies
$$h(\{k\})\geq \frac{\alpha_k}{p_k}$$
\end{theorem}
\begin{proof}
In order to prove the theorem, it suffices to show that for any $S\subseteq [k]$ such that $S\neq\emptyset$ and $i\in S$, it holds that
$$p_i(h(S)-h(S\setminus\{i\})) \geq \alpha_{k-|S|+1}$$
Then, by setting $S=\{k\}$ ($|S|=1$) and $i=k\in S$, we get $p_k(h(\{k\})-h(\emptyset)) \geq \alpha_k$, and since $h(\emptyset)=0$ by definition, the Theorem follows. It remains to prove the desired property. This can be shown by induction on the size of $S$.

\textit{Base case}: If $|S|=k$ (this means that $S=[k])$ then for any $i\in S$, by \eqref{system:lb} we have $$p_i(h(S)-h(S\setminus \{i\})) = 1 = \alpha_1 = \alpha_{k-|S|+1}.$$

\textit{Inductive hypothesis}: Suppose that for any $S\subseteq [k]$ with $|S|=\ell>1$ and any $i\in S$, we have $$p_i(h(S)-h(S\setminus \{i\}))\geq \alpha_{k-\ell+1}.$$   

\textit{Inductive step}: Let $S\subseteq [k]$ be any set with $|S|=\ell-1>0$ and $i\in S$ be any element of this set. By Lemma \ref{lemma:induction_property}, we have that $$p_i(h(S)-h(S\setminus \{i\})) \geq 1 + \sum_{j\notin S}p_m(h(S\cup\{j\})-h(S))$$
Now, for any $j\notin S$ we can use the hypothesis on the set $S\cup\{j\}$ with size $\ell$. Thus, we have
$$p_j(h(S\cup\{j\})-h(S)) \geq \alpha_{k-\ell+1} = \alpha_{k-|S|}$$
for any $j\notin S$. Combining, we get
$$p_i(h(S)-h(S\setminus \{i\})) \geq 1 + (k-|S|)\alpha_{k-|S|} = \alpha_{k-|S|+1}.$$
\end{proof}

\paragraph{Proof of Theorem~\ref{thm:gen_lower_bound}.} Since $p_1\geq p_2 \geq \dots \geq p_k$, we have that $p_k\leq \frac{1}{k}$. Thus, by Theorem~\ref{th:h(k)_lb} we have that $h(\{k\})\geq k\cdot\alpha_k$ for any distribution. Since $h(\{k\}$ is a lower bound for any memoryless algorithm, the Theorem follows.

\begin{corollary}\label{cor:harmonic_is_best}
The Harmonic Algorithm is the only memoryless algorithm with a competitive ratio of $k\cdot\alpha_k$.
\end{corollary}
\begin{proof}
By Theorem~\ref{th:h(k)_lb}, the competitive ratio of the Harmonic Algorithm is at least $k\cdot\alpha_k$ and combined with the upper bound of Theorem~\ref{thm:upper_bound} we get that the Harmonic Algorithm is ($k\cdot\alpha_k$)-competitive. Assuming $p_1\geq \dots \geq p_k$, any other memoryless algorithm will have $p_k<\frac{1}{k}$. Thus, by Theorem~\ref{th:h(k)_lb} its competitive ratio will be lower bounded by $h(\{k\})>k\cdot \alpha_k$ which is strictly worse that the competitive ratio of the Harmonic Algorithm.
\end{proof}

%% file: 04-extensions.tex
\section{Concluding Remarks}
\label{sec:ext}

We provided tight bounds on the competitive ratio of randomized memoryless algorithms for generalized $k$-server in uniform metrics. Combining our results with the work of Chiplunkar and Vishwanathan~\cite{ChV20}, the power of memoryless algorithms in uniform and weighted uniform metrics is completely characterized. It might be interesting to determine the power of memoryless algorithms for other metric spaces such as e.g., weighted stars. However we note that memoryless algorithms are not competitive on arbitrary metric spaces, even for $k=2$; this was shown by Chrobak and Sgall~\cite{CS04} and Koutsoupias and Taylor~\cite{KT04} independently. We conclude with some side remarks.

\paragraph{Metrics with $n=2$ points.} In our lower bound instance from Section~\ref{sec:gen_lb} we require that all metric spaces have at least $n \geq 3$ points. We observe that this is necessary, and that if all metric spaces have $n=2$ points, the Harmonic Algorithm is $O(2^k)$-competitive, thus a lower bound of $k \cdot \alpha_k$ can not be achieved. The underlying reason is the following: in the Harmonic Chain described in Section~\ref{section:harmonic_upper_bound}, while being at state $\ell$ (i.e., having $\ell$ servers different than the adversary), the algorithm moves to state $\ell-1$ with probability $1/k$ and remains in the same state with probability  $ \frac{\ell-1}{k} $. This happens because if $n\geq 3$,
then given the algorithm's configuration $q$ and the adversary's configuration $a\neq q$, we can construct a request $r$ such that $r_i\neq q_i$ and $r_i\neq a_i$ in $\ell-1$ metric spaces. However if $n=2$, the algorithm moves only for $r=\bar{q}$ (i.e., $r$ is the algorithm's anti-configuration) and thus $a_i\neq q_i$ implies that $r_i=a_i$ and if the algorithm moves in $M_i$, then it reduces the number of different servers to $\ell-1$. Thus the Markov Chain used to analyse this instance becomes the following:

\begin{figure}[ht]
    \centering
    \includegraphics[width=\textwidth]{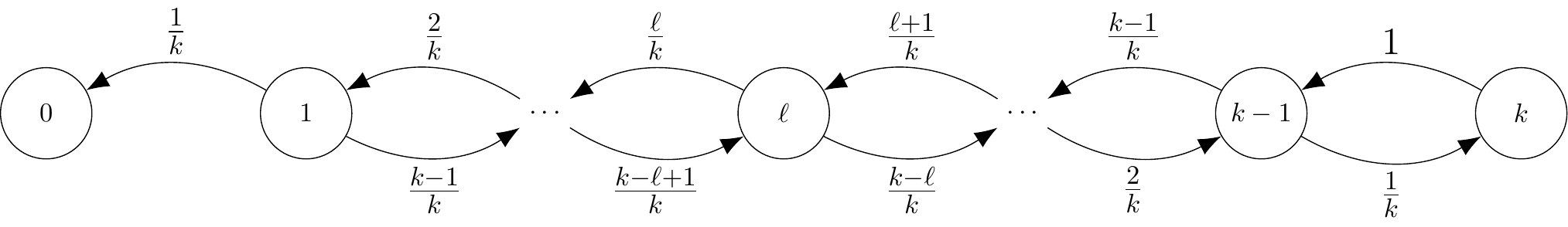}
    \caption{The evolution of the Harmonic Algorithm's Hamming distance from the adversary when $n=2$.}
\end{figure}

Then, as we show in Appendix~\ref{app:harmonic_n=2}, for this random walk, $h(\ell) = O(2^k)$ for any $1\leq \ell \leq k$ and using a similar technique as in Section~\ref{section:harmonic_upper_bound} we can prove that the Harmonic Algorithm is $O(2^k)$-competitive.

\paragraph{Randomized Algorithms with Memory.} We note that for uniform metrics, if memory is allowed and we compare against oblivious adversaries, competitive randomized algorithms are known: Bansal et. al.~\cite{BEKN18} designed a $O(k^3 \log k)$-competitive randomized algorithm with memory; this was recently improved to $O(k^2 \log k)$ by Bienkowski et. al.~\cite{BJS19}. 

%% file: A-markov-chains.tex
\section{Analysis of the Harmonic Chain (Proof of Theorem~\ref{thm:harmonic_EET})}\label{app:markov_chains}

In this part of the Appendix, our main objective is to prove Theorem~\ref{thm:harmonic_EET} that states the EET of the \textit{Harmonic Chain}, which is a special type of Markov Chain that we use in our analysis. We note that a family of Markov Chains of similar structure, called \textit{Birth-Death Chains}, has been extensively studied in the literature~\cite{birth_death_book}.

\begin{definition}[Birth-Death Chain]
A \textit{Birth-Death Markov Chain} is an important sub-class of discrete-time Markov Chains that limits transitions to only adjacent states. Formally, a Markov process with state-space $\pazocal{X} = \{0,1, \dotsc ,k\}$ for some $k\in\N$ is characterized as a Birth-Death Chain if its transition matrix $P = [P_{ij}]$ has the following form:

\[ 
P_{ij}= \left\{
\begin{array}{ll}
      p_i & ,j=i+1 \\
      q_i & ,j = i-1\\
      1 - p_i - q_i &,j=i \\
      0 & \text{,otherwise} \\
\end{array}  \quad \quad \forall i,j \in \pazocal{X}
\right. 
\]
where $q_0 = 0$ and $p_k = 0$ for the end-points of the chain. A graphical representation of a Birth-Death chain is given in Figure \ref{fig:bdchain}.
\end{definition}

\begin{figure}[ht]
    \centering
    \includegraphics[width=\textwidth]{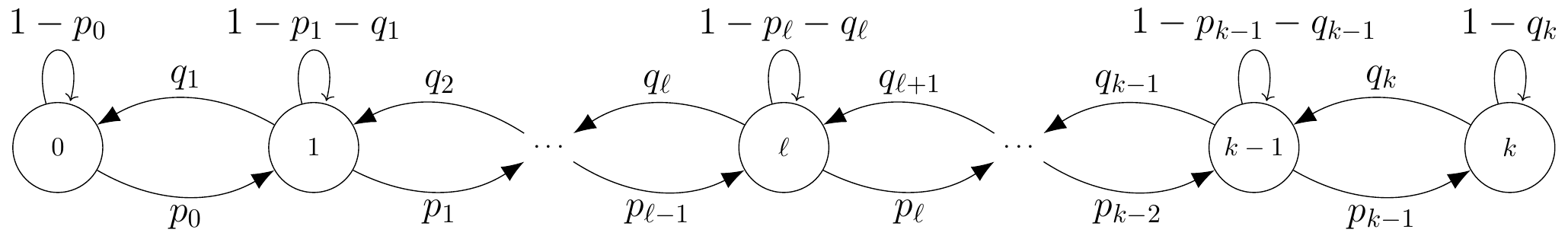} 
    \caption{A Birth Death Markov Chain} \label{fig:bdchain}
\end{figure}
\FloatBarrier

Furthermore, a Birth Death Chain will be called \textit{absorbing} on the state $X=0$ if $p_0 = 0$, which means that the random process will remain on the state $X=0$ if it ever reaches it. As we mentioned, the Harmonic Chain is a special case of Birth-Death Chains. A formal definition is given below:

\begin{definition}[Harmonic Chain]
The Harmonic Chain can be defined as a Birth-Death Chain with state-space $\pazocal{X}=\{0,1,\dotsc, k\}$ for some $k\in\N$, forward probabilities $q_i = \frac{1}{k}$ and backward probabilities $p_i = \frac{k-i}{k}$. A graphical representation of a Harmonic Chain is given in Figure \ref{fig:harchain}.
\end{definition} 
\begin{figure}[ht]
    \centering
    \includegraphics[width=\textwidth]{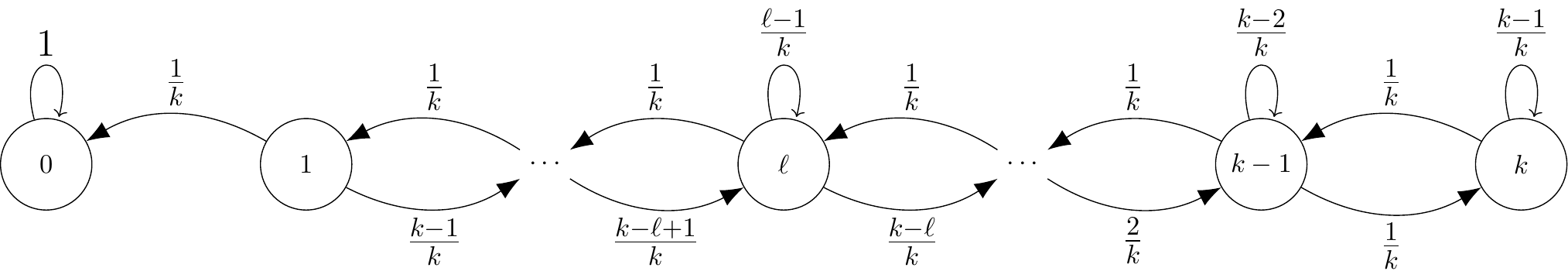} 
    \caption{The Harmonic Chain} \label{fig:harchain}
\end{figure}
\FloatBarrier

We will compute the \textit{Expected Extinction Time} (EET) of a Birth-Death Chain starting from an initial state $\ell\in \pazocal{X}$, which is defined as the expected number of transitions needed to reach $X=0$ for the first time, starting from state $\ell$. Formally, the EET of a Birth-Death chain starting from some state $X_0=\ell\in\pazocal{X}$ is defined as $h(\ell) = \mathbb{E}[N|X_0=\ell]$ where the random variable $N$ is defined as $N=\min_{\tau \geq 0}\{X_\tau=0\}$. A closed-form expression for the EET $h(\ell)$ of a Birth-Death Chain is given by the following Theorem:

\begin{theorem}\label{th:eet_of_bdchain}
For any Birth-Death Chain with states $\pazocal{X}=\{0,1,\dotsc, k\}$ for some $k\in\N$, transition probabilities $p_i,q_i$ and absorption state $X=0$, the EET starting from an initial state $X_0=\ell\in\{1,2,\dotsc, k\}$ is given by
$$h(\ell) = h(1) + \sum_{i=1}^{\ell-1}(\frac{q_1\cdots q_i}{p_1 \cdots p_i}\sum_{j=i+1}^{k}\frac{p_1\cdots p_{j-1}}{q_1 \cdots q_{j}}),$$
where
$$h(1) = \frac{1}{q_1} + \sum_{i=2}^{k} \frac{p_1 \cdots p_{i-1}}{q_1\cdots q_i}.$$
\end{theorem}

While this result is by no means novel, a Birth-Death Chain is usually defined on a state space $\pazocal{X}'=\{0,1,\dotsc\}$ in the literature, while we study Birth-Death Chains with (finite) state space $\pazocal{X}=\{0,1,\dotsc,k\}$. Thus, for the sake of completion, we give a formal proof of Theorem~\ref{th:eet_of_bdchain} in~\ref{subapp:birth_death_chains}.

\paragraph{Proof of Theorem~\ref{thm:harmonic_EET}} With Theorem~\ref{th:eet_of_bdchain} stated, we are now ready to prove Theorem~\ref{thm:harmonic_EET}. We have defined the the Harmonic Chain as a Birth-Death Chain with forward probabilities $q_i=\frac{1}{k}$ and backward probabilities $p_i=\frac{k-i}{k}$. From these probabilities and Theorem~\ref{th:eet_of_bdchain}, it is simple to compute that for any $\ell\in\{1,2,\dots,k\}$, the EET of the Harmonic Chain is given by
\begin{equation}\label{eq:h(l)_not_ready}
h(\ell) = k + k!\sum_{j=0}^{k-2}\frac{1}{j!} + k\sum_{i=k-\ell}^{k-2}i!\sum_{j=0}^i\frac{1}{j!}
\end{equation}

Recall that by definition, $\alpha_\ell = 1 + (\ell-1)\alpha_{\ell-1}$ with $\alpha_1=1$. As we noted in~\eqref{eq:harmonic_recursion_closed_form}, this recursion has a closed form given by
$\alpha_\ell= (\ell-1)!\sum_{i=0}^{\ell-1}\frac{1}{i!}$. Using this to substitute the sums of inverse factorials in~\eqref{eq:h(l)_not_ready}, we get that
$$h(\ell) = k + k!\frac{\alpha_{k-1}}{(k-2)!} + k\sum_{i=k-\ell}^{k-2}\alpha_{i+1} = k\sum_{i=k-\ell+1}^{k}\alpha_i$$
and conclude the proof of Theorem~\ref{thm:harmonic_EET}.


\subsection{Analysis of the Birth-Death Chain (Proof of Theorem~\ref{th:eet_of_bdchain})}\label{subapp:birth_death_chains}

Obviously, for the EET of the absorbing state $X=0$ we have $h(0) = 0$ be definition. For any other initial state $\ell\in\{1,2,\dotsc, k\}$, by using conditional probabilities on the definition of the EET, we get
\begin{equation}\label{eq:rec_rel}
h(\ell) = 1 + q_\ell h(\ell-1) + p_\ell h(\ell+1) + (1-p_\ell-q_\ell)h(\ell)
\end{equation}

This is a second-order recurrence relation we need to solve in order to compute the EET of a Birth-Death Chain. There are two main points in the proof. Firstly, by solving the equations on the differences $h(\ell)-h(\ell-1)$ instead on $h(\ell)$, we can reduce the problem to solving a first-order recurrence relation that is easier to solve. Secondly, a second-order recurrence relation generally needs two initial conditions $h(0),h(1)$ in order to be solved. Notice that we only know $h(0)=0$. Using well known results from the literature on Markov Chains, we show how to compute $h(1)$ and overcome this technical problem. 

By rearranging~\eqref{eq:rec_rel}, we get

$$h(\ell+1) - h(\ell) = \frac{q_\ell}{p_\ell}(h(\ell)-h(\ell-1) - \frac{1}{q_\ell}), \quad \forall \ell = 1,2,\dots, k$$
with $h(0)=0$. For $\ell=0,1\dots,k-2$, we define $\Delta(\ell) = h(\ell+2) - h(\ell+1)$ and get:

$$\Delta(\ell+1) = f_\ell\Delta(\ell) + g_\ell, \;\; \Delta(0) = \frac{q_1}{p_1}(h(1) - \frac{1}{q_1})$$
where $f_\ell = \frac{q_{\ell+2}}{p_{\ell+2}}$ and $g_\ell = -\frac{1}{p_{\ell+2}}$. This is a first-order non-homogeneous recurrence relation with variable coefficients that yields the solution

$$\Delta(\ell) = (\prod_{i=0}^{\ell-1}f_i)(\Delta(0) + \sum_{i=0}^{\ell-1}\frac{g_i}{\prod_{j=0}^i f_j})$$

Finally, by substitution of $f_i$, $g_i$ and $\Delta(0)$ and by using the telescoping property we have $h(\ell) = h(1) + \sum_{i=0}^{\ell-2}\Delta(i)$ to solve for $h(\ell)$, we get that for any $\ell\in\{1,2,\dotsc,k\}$:
\begin{equation}\label{eq:h(l)_from_h(1)}
h(\ell) = h(1) + \sum_{i=1}^{\ell-1}\frac{q_1\cdots q_i}{p_1 \cdots p_i}(h(1) - \frac{1}{q_1} - \sum_{j=2}^{i}\frac{p_1 \cdots p_{j-1}}{q_1 \cdots q_{j}})
\end{equation}

It remains to determine the value of $h(1)$. Notice that if we set $p_0 = 1$ (i.e. state $0$ always transitions to state $1$ instead of being absorbing) the EET of the chain won't change, since we are only interested in the transitions until state $X=0$ is reached for the first time. However, if $p_0 = 1$ and $T_0$ denotes \textit{return time} of state $X=0$ (i.e., the time it takes to return to $X=0$ starting from $X=0$), it holds that 
\begin{equation}\label{eq:e(t0)}
E(T_0) = h(1) + 1
\end{equation}

For Markov Chains, it is known~\cite{prob:2003} that the Expected Return Time $E(T_\ell)$ of a state $\ell\in\pazocal{X}$ is given by $E(T_\ell)=\frac{1}{\pi_\ell}$ where $\pi$ is the stationary distribution of the Markov Chain. The stationary distribution $\pi$ of a Birth-Death Chain is known to satisfy
$$\pi_\ell p_\ell = \pi_{\ell+1}q_{\ell+1} \Rightarrow \pi_{\ell+1} = \frac{p_\ell}{q_{\ell+1}}\pi_\ell \Rightarrow \pi_{\ell} = \frac{p_{\ell-1} p_{\ell-2} \cdots p_0 }{q_{\ell}q_{\ell-1}\cdots q_1} \pi_0 \quad\;\;\forall \ell \in \{1,2,\dotsc, k\}$$

To see this, imagine "cutting" the chain between two adjacent states $\ell,(\ell+1)\in\pazocal{X}$. As the random process evolves $(n\rightarrow \infty$), the probability of being in a state $\ell\in\pazocal{X}$ is given by the stationary distribution $\pi_\ell$ of the chain and the number of times that the cut is crossed from the left to the right must become equal to the number of times that the cut is crossed from the right to the left. Finally, since $\sum_{i=0}^{k}\pi_i = 1$ and $p_0=1$, we get
\begin{equation}\label{eq:stationary_at_0}
\pi_0 = \frac{1}{1 + \sum_{i=1}^{k} \frac{p_1 \cdots p_{i-1}}{q_1\cdots q_i}}
\end{equation}

Combining equations~\eqref{eq:h(l)_from_h(1)},~\eqref{eq:e(t0)} and~\eqref{eq:stationary_at_0}, Theorem~\ref{th:eet_of_bdchain} follows.

%% file: B-Chiplunkar.tex
\section{Proof of Lemma \ref{lemma:monotonicity}}\label{app:Chiplunkar}

In this part of the Appendix, we will formally give the proof for Lemma~\ref{lemma:monotonicity}. As we mentioned in Section~\ref{sec:gen_lb}, linear system~\eqref{system:lb} is equivalent to the linear system studied by Chiplunkar and Vishnawathan~\cite{ChV20} in order to show tight doubly-exponential bounds for the generalization of the problem where the metrics have unequal weights. However, as our analysis in Section~\ref{subsection:general_lb_equations} showed when the metrics have equal weights, this bound reduces to $\Theta(k!)$.

In order to prove Lemma~\ref{lemma:monotonicity}, we will first show that~\eqref{system:lb} is indeed equivalent to the equations studied in~\cite{ChV20}. Recall that linear system~\eqref{system:lb} is defined as:

\begin{equation*}
\left\{
\begin{array}{c}
     h(\emptyset) = 0\\
     p_m(h(S)- h(S\setminus \{m\})) = 1 +\sum_{j\notin S}p_j(h(S\cup \{j\})-h(S)), \quad \forall S\neq\emptyset , m=\min (S)\\
\end{array} 
\right\} 
\end{equation*}
under the assumption that $p_1\geq p_2 \geq \dots \geq p_k$. By introducing a new set of variables $\phi$, defined as
\begin{equation}\label{eq:phi_to_h}
    \phi(S) = h([k]) - h([k]\setminus S),\;\;\forall S\subseteq [k]
\end{equation}
we get that $\phi(\emptyset)=0$ and $\forall S\neq\emptyset$ with  $m=\min(S)$ we get 
$$p_m(\phi(([k]\setminus S)\cup\{m\})-\phi([k]\setminus S)) = 1 +\sum_{j\notin S}p_j(\phi([k]\setminus S)-\phi([k]\setminus S\setminus\{j\}))$$

Lastly, by re-writing the equations using $\bar{S}=[k]\setminus S$, we end up with the following (equivalent) linear system:

\begin{equation}\label{eq:their_system}
\left\{
\begin{array}{c}
     \phi(\emptyset) = 0\\
     p_m(\phi(\bar{S}\cup\{m\})-\phi(\bar{S}) = 1 +\sum_{j\in \bar{S}}p_j(\phi(\bar{S})-\phi(\bar{S}\setminus\{j\})), \quad \forall \bar{S}\neq [k] ,\; m=\min([k]\setminus \bar{S})\\
\end{array} 
\right\} 
\end{equation}

We remark that this is the exact set of equations studied by Chiplunkar and Vishnawathan (see equations (6),(7) in~\cite{ChV20}). Using the Gauss-Seidel Trick technique, they prove that the solutions of~\eqref{eq:their_system} (and thus the solutions of~\eqref{system:lb}) always exist. Then, in Lemma 3.3 of their paper they proved the following property for the solutions of linear system \eqref{eq:their_system}:

\begin{lemma}[\textbf{Lemma 3.3 of Chiplunkar et. al.}]
For any $\bar{S}\subseteq [k]$ with $i,j\notin \bar{S}$ and $i<j$ (thus $p_i\geq p_j$), the solution of linear system \eqref{eq:their_system} satisfies
$$p_i(\phi(\bar{S}\cup\{i\}) - \phi(\bar{S})) \leq p_j(\phi(\bar{S}\cup\{j\}) - \phi(\bar{S}))$$
\end{lemma}

By re-writing this lemma for $S=[k]\setminus \bar{S}$, we get that for any $S\subseteq [k]$ with $i,j\in S$ and $i<j$, the solution of linear system \eqref{eq:their_system} satisfies
$$p_i(\phi([k]\setminus (S\setminus\{i\})) - \phi([k]\setminus S)) \leq p_j(\phi([k]\setminus (S\setminus\{j\})) - \phi([k]\setminus S))$$

Lastly, by equation \eqref{eq:phi_to_h} we get that for any $S\subseteq [k]$ with $i,j\in S$ and $i<j$, the solution of linear system \eqref{system:lb} satisfies
$$p_i(h(S)-h(S\setminus\{i\})) \leq p_j(h(S)-h(S\setminus\{i\}))$$
that is, we have that Lemma \ref{lemma:monotonicity} holds.



%% file: C-n=2.tex
\section{Bounds in 2-point metric spaces}\label{app:harmonic_n=2}
In this part of the Appendix, we analyze the Markov Chain that captures the evolution of the Hamming distance between the Harmonic Algorithm and the adversary in metric spaces with $n=2$ points. We study the following Markov Chain:

\begin{figure}[ht]
    \centering
    \includegraphics[width=\textwidth]{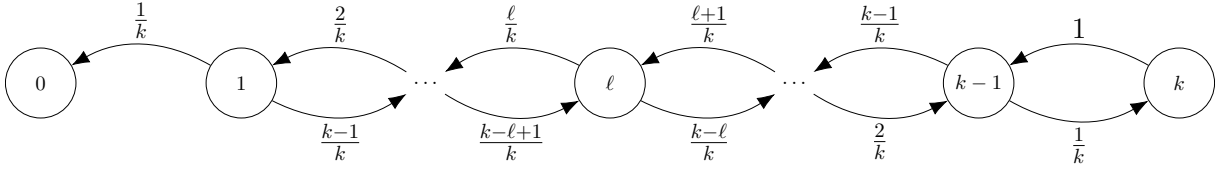}
    \caption{The evolution of the Harmonic Algorithm's Hamming distance from the adversary when $n=2$.}
    \label{figure:binary}
\end{figure}

Notice that this is a Birth-Death Chain (defined in Appendix~\ref{app:markov_chains}) with probabilities $q_\ell=\frac{\ell}{k}$ and $p_\ell = \frac{k-\ell}{k}$. Using Theorem~\ref{th:eet_of_bdchain} and these probabilities, we get that for any $\ell\in\{1,2,\dotsc,k\}$
\begin{equation}\label{eq:binary}
h(\ell) =  2^k-1 + \sum_{i=1}^{\ell-1}\frac{1}{\binom{k-1}{i}}(  2^k -  \sum_{j=0}^i\binom{k}{j})
\end{equation}

\begin{theorem}\label{thm:binary_EET}
In the Markov Chain of Figure~\ref{figure:binary}, for any $\ell\in\{1,2,\dotsc ,k\}$, it holds that the EET is $h(\ell)=\Theta(2^k)$.
\end{theorem}
\begin{proof}
Observe that for any $\ell\in\{1,\dotsc,k\}$ and $i\leq \ell-1$, we have $\sum_{j=0}^i\binom{k}{j} \leq 2^k$ and thus $h(\ell)\geq 2^k-1$. Also,

\begin{equation*} 
\begin{split}
\sum_{i=1}^{\ell-1}\frac{1}{\binom{k-1}{i}}(2^k - \sum_{j=0}^i\binom{k}{j}) & \leq 2^k\sum_{i=1}^{\ell-1}\frac{1}{\binom{k-1}{i}} \\
& \leq 2^k\sum_{i=1}^{k-1}\frac{1}{\binom{k-1}{i}} \\
& = 2^k(1 + \frac{2}{k-1} + \sum_{i=2}^{k-3}\frac{1}{\binom{k-1}{i}}) \\
& \leq 2^k(3 + \frac{k-4}{\binom{k-1}{2}}) \\
& = 2^k(3 + \frac{2(k-4)}{(k-2)(k-1)}) = O(2^k)
\end{split}
\end{equation*}
where the last inequality follows from the fact that $\binom{\ell}{i} \geq \binom{\ell}{2}$ for any $i\in [2,\ell-2]$. Combining with~\eqref{eq:binary}, we get that $h(\ell)=\Theta(2^k)$ for any $\ell \in\{1,\dotsc,k\}$.
\end{proof}

Using similar techniques to those of Section~\ref{section:harmonic_upper_bound}, Theorem~\ref{thm:binary_EET} implies an $O(2^k)$ upper bound for the competitive ratio of the Harmonic Algorithm.